\documentclass[11pt]{article}
\usepackage[T1]{fontenc}
\usepackage{lmodern}
\usepackage[margin=1in]{geometry}
\usepackage{amsmath,amssymb,amsthm,mathtools}
\usepackage{booktabs,microtype,needspace}
\usepackage[numbers,sort&compress]{natbib}
\usepackage{xurl}
\usepackage[colorlinks=true,linkcolor=blue!50!black,citecolor=blue!50!black,urlcolor=blue!50!black]{hyperref}
\usepackage{xcolor}
\hypersetup{pdftitle={A faster matrix multiplication algorithm through structured optimization},pdfauthor={Reza Zadeh},pdfsubject={Faster asymptotic matrix multiplication through conditioned optimization and exact certification},pdfkeywords={matrix multiplication algorithm, exponent, laser method, combination loss, rational certificate}}
\newtheorem{theorem}{Theorem}
\newtheorem{lemma}{Lemma}
\newtheorem{proposition}{Proposition}
\theoremstyle{remark}

\title{A faster matrix multiplication algorithm\\through structured optimization}
\author{Reza Zadeh\\Matroid, Inc.\\\texttt{reza@matroid.com}}
\date{October 6, 2026}

\begin{document}
\maketitle
\begin{abstract}
We present a new asymptotic matrix multiplication construction yielding
$\omega<2.37115924$, and therefore an algorithm that multiplies two
$n\times n$ matrices in $O(n^{2.37115924})$ arithmetic operations over
the rationals, reals, or complex numbers. The construction improves the
bound $2.371177$ of Dupont et al. within the combination-loss framework
for the Coppersmith--Winograd tensor. Our approach combines a shared
recursive parameter family with a conditioned optimization procedure
that addresses two obstacles: weak gradients at low-mass constituents
and competing bottlenecks in the retained exponent. A mass- and
probability-dependent rescaling supports quasi-Newton optimization in
$212\,820$ coordinates, while a six-control correction balances the
three extraction roles at each relevant stage. Positive rational
entropy-dual factors and an independent ordinary-string verifier turn
the resulting construction into an exact certificate. The verified
upper-bound expression lies in
$[2.3711592385931075066123,\,2.3711592385931075066124]$.
The complete witness and an offline verifier using only the Python
standard library accompany the paper.
\end{abstract}

\section{Introduction}

We give a new asymptotic algorithm for matrix multiplication, reducing
the exponent bound to $2.37115924$. The improvement arises from a
specific recursive tensor construction discovered by an optimization
procedure tailored to the geometry and bottlenecks of combination-loss
analysis. The resulting rational construction is accompanied by an
independently implemented exact verifier.

The matrix multiplication exponent $\omega_{\mathbb Q}$ is the infimum
of the real numbers $a$ such that two $n\times n$ matrices over
$\mathbb Q$ can be multiplied using $O(n^{a+\varepsilon})$ arithmetic
operations for every $\varepsilon>0$. We write $\omega=\omega_{\mathbb Q}$
throughout. Since Strassen's subcubic algorithm~\cite{strassen1969},
successive improvements have relied on algebraic constructions that
reduce many matrix products to tensor decompositions. The
Coppersmith--Winograd tensor~\cite{coppersmith1990}, together with
the asymptotic sum inequality~\cite{schonhage1981}, provides the
starting point for the construction studied here.

Combination-loss analysis refines the laser method by retaining
additional combinations of tensor constituents. Its recent developments
include asymmetric hashing~\cite{duan2023}, refined bounds for square
and rectangular multiplication~\cite{williams2024}, and the asymmetric
framework of Alman et al.~\cite{alman2025}. These theorems reduce an
upper-bound calculation to a constrained optimization problem involving
probability distributions, entropy penalties, and recursively derived
matrix dimensions. A feasible parameter choice is sufficient; global
optimality is unnecessary.

Dupont et al.~\cite{dupont2026} formulate this optimization at recursion
level four and obtain $\omega<2.371177$. The central question addressed
here is how to obtain a better algorithmic construction from that
framework. We combine parameter sharing with mass-dependent
conditioning and explicit correction of the extraction bottlenecks.
The resulting construction proves the following algorithmic guarantee.

\Needspace{9\baselineskip}
\begin{theorem}
\label{thm:main}
There is an algorithm for multiplying two $n\times n$ matrices over
$\mathbb Q$ using $O(n^{2.37115924})$ arithmetic operations.
More precisely, the square matrix multiplication exponent satisfies
\[
\omega_{\mathbb Q}<2.37115924.
\]
The same algorithmic guarantee holds over every characteristic-zero extension
of $\mathbb Q$, including $\mathbb R$ and $\mathbb C$.
\end{theorem}

Relative to the bound of Dupont et al., the reduction in the exponent is
$1.776\cdot10^{-5}$. This is an improvement in the polynomial growth
rate of the arithmetic complexity: the corresponding leading powers
of $n$ differ by a factor $n^{0.00001776}$. As with other algorithms
obtained through this tensor framework, the guarantee is asymptotic;
the exponent alone does not specify a practical crossover size.
A contemporaneous public computational
release~\cite{publiccertificate2026} reports an upper-bound expression
approximately equal to $2.371169352063661$; the present certified
expression is also smaller than that value.

Table~\ref{tab:historical-exponents} places the result among selected
major improvements in the matrix multiplication exponent.

\begin{table}[t]
\centering
\small
\renewcommand{\arraystretch}{1.12}
\begin{tabular}{@{}llr@{}}
\toprule
Year & Work and reference & Bound on $\omega$\\
\midrule
1969 & Strassen~\cite{strassen1969} & $\leq\log_2 7<2.807355$\\
1978 & Pan~\cite{pan1978} & $<2.796$\\
1979 & Bini, Capovani, Romani, and Lotti~\cite{bini1979} & $<2.7799$\\
1981 & Sch\"onhage~\cite{schonhage1981} & $<2.522$\\
1982 & Romani~\cite{romani1982} & $<2.517$\\
1982 & Coppersmith and Winograd~\cite{coppersmith1982} & $<2.495548$\\
1986 & Strassen~\cite{strassen1986} & $<2.479$\\
1990 & Coppersmith and Winograd~\cite{coppersmith1990} & $<2.375477$\\
2010 & Stothers~\cite{stothers2010} & $<2.37369$\\
2012 & Vassilevska Williams~\cite{williams2012} & $<2.3729$\\
2014 & Le Gall~\cite{legall2014} & $<2.3728639$\\
2021 & Alman and Vassilevska Williams~\cite{alman2021} & $<2.3728596$\\
2023 & Duan, Wu, and Zhou~\cite{duan2023} & $<2.371866$\\
2024 & Vassilevska Williams et al.~\cite{williams2024} & $\leq2.371552$\\
2025 & Alman et al.~\cite{alman2025} & $<2.371339$\\
2026 & Dupont et al.~\cite{dupont2026} & $<2.371177$\\
2026 & BSD (public certificate)~\cite{publiccertificate2026} & $<2.371169353$\\
\midrule
2026 & \textbf{This paper} (Theorem~\ref{thm:main}) & $<\mathbf{2.37115924}$\\
\bottomrule
\end{tabular}
\caption{Selected improvements in upper bounds on the matrix
multiplication exponent. Years refer to the cited publications;
some results appeared earlier in proceedings or preprints. The
Vassilevska Williams entry uses a bound from the corrected full
version. The BSD row is a public computational certificate, rather
than a paper, and its bound is rounded upward from the exact
rational certificate.}
\label{tab:historical-exponents}
\end{table}

\paragraph{Contributions.}
The paper contributes a new certified matrix multiplication
construction and a concrete combination of techniques for finding and
verifying it:
\begin{enumerate}
\item \textbf{A faster asymptotic algorithm from a compact construction.}
The rational witness specifies an instance of the level-four
combination-loss construction with $q=5$. Shared descendant parameters
and binary-tree position symmetries give a $212\,820$-coordinate family
in which the improved exponent is attained. Unfolding the shared graph
preserves admissibility, and six global coordinate permutations produce
equal final matrix dimensions.
\item \textbf{Conditioning adapted to constituent masses.}
We rescale each numerical coordinate using its occurrence mass and
local probability, including marginal-variance weights for entropy-dual
coordinates. This addresses the disparate gradient scales of the
recursive construction. The rescaling is refreshed between limited-memory
quasi-Newton solves, with a larger scale range for low-mass constituents.
\item \textbf{A six-control correction for extraction bottlenecks.}
The root and two constituent stages each contribute a minimum of
three role totals. We couple an augmented-Lagrangian continuation to
a separate correction using two controls per stage. A $6\times6$
Newton system targets the six role differences, and candidate selection
uses the original minimum-based objective. This focuses the correction
on the bottlenecks that directly determine the exponent.
\item \textbf{An exact, independently checkable construction.}
The export preserves positive entropy-dual factors as arbitrary-size
rational numbers. An independent verifier expands ordinary ternary
strings and assembles all physical-coordinate contributions before
applying the published minima. The resulting finite certificate can
be checked offline, separately from the optimization implementation.
\end{enumerate}

These contributions build on the tensor theorem of Alman et al.
and the optimization formulation of Dupont et al. The orbit
representation and equality-constrained approach were informed by
computational notes shared by Eric Naslund. Softmax parameterization,
quasi-Newton optimization, augmented Lagrangians, entropy duality, and
rational interval arithmetic are established tools; the methodological
contribution is their specified adaptation and integration for the
construction certified here.

Section~\ref{sec:construction} specifies the construction and the
restricted parameter family. Section~\ref{sec:optimization} describes
the numerical continuation and rational export.
Section~\ref{sec:verification} gives the certificate and the proof of
Theorem~\ref{thm:main}. The appendix supplies implementation-level
lemmas for shared parameters and logarithm enclosures.

\section{Construction and parameter family}
\label{sec:construction}

The construction is an instance of the combination-loss framework of
\cite{alman2025}, in the formulation of \cite[Section~2]{dupont2026}.
The new algorithm is specified by a rational recursive extraction
construction within that framework. Its parameters determine the
retained matrix products and their asymptotic dimensions.  We use the Coppersmith--Winograd tensor with parameter $q=5$ and
recursion level $\ell^*=4$.  Thus the starting tensor is
$\mathrm{CW}_5^{\otimes 8}$, whose asymptotic rank cost, in binary logarithms,
is
\begin{equation}
 C=2^{\ell^*-1}\log_2(q+2)=8\log_2 7.
 \label{eq:rank-cost}
\end{equation}
Below we specify the restricted family used in the numerical search.  All
restrictions are parameter identifications or symmetries allowed by the
published construction.

\subsection{Shapes, complete split distributions, and recursive parameters}

A level-$\ell$ shape is a triple $s=(i,j,k)$ of nonnegative integers with
$i+j+k=2^\ell$.  For each coordinate $w\in\{X,Y,Z\}$, its complete split
distribution $\beta_s^w$ is supported on
\[
 \mathcal S_{\ell,s_w}
 =\left\{L\in\{0,1,2\}^{2^{\ell-1}}:
                  \sum_t L_t=s_w\right\}.
\]
Different occurrences of the same shape may have different parameters;
we suppress the occurrence index when it is unambiguous.  The root
distribution $\alpha^*$ is supported on all $153$ shapes of level~$4$.

For a positive shape $s$ at level $\ell\geq 3$, define its allowed
half-shapes by
\[
 D_s=\{u\in\mathbb Z_{\geq0}^3:
             u_X+u_Y+u_Z=2^{\ell-1},\ u_w\leq s_w\text{ for all }w\}.
\]
There are six extraction regions, indexed by permutations
$p_r:\{X,Y,Z\}\to\{X,Y,Z\}$.  At an occurrence $v$ of shape $s$, the
parameters comprise a probability vector $(A_{v,r})_{r=1}^6$, a distribution
$\alpha_{v,r}$ on $D_s$ for each region, and references $c(v,r,u)$ to
children of shape $u$.  The regional complete split distributions are
\begin{equation}
 B_{v,r}^w
 =\sum_{u\in D_s}\alpha_{v,r}(u)
       \bigl(\beta_{c(v,r,u)}^w
                   \otimes\beta_{c(v,r,s-u)}^w\bigr),
 \qquad
 \beta_v^w=\sum_{r=1}^6 A_{v,r}B_{v,r}^w,
 \label{eq:beta-recursion}
\end{equation}
where the tensor product denotes ordered concatenation of the two split
strings.  In the family used for the final witness,
\begin{equation}
 \alpha_{v,r}(u)=\alpha_{v,r}(s-u).
 \label{eq:complement-symmetry}
\end{equation}
The root distribution $\alpha^*$ has no corresponding constraint under
permutation of its three coordinates.

For a shape with a zero coordinate, one nonzero coordinate's complete
split distribution is free.  The other is its image under
$L\mapsto\mathbf 2-L$, and the zero coordinate is concentrated on the
all-zero string.  A positive shape at level~$2$ is a permutation of
$(1,1,2)$.  It has a parameter $\mu\in[0,1/2]$: each coordinate of value
$1$ has probabilities $(1/2,1/2)$ on $(01,10)$, and the coordinate of
value $2$ has probabilities $(\mu,1-2\mu,\mu)$ on $(02,11,20)$.

The search imposes two further identifications.  First, complete split
distributions are invariant under the group generated by exchanging the
two children of each node in a complete binary tree on the string
positions.  Only one probability per orbit is stored, with the orbit
multiplicity included in normalization and entropy evaluation.  Equation
\eqref{eq:complement-symmetry} makes the recursive product construction
compatible with exchanging the two halves.  Second, at the transition
from level~$4$ to level~$3$, children with the same parent occurrence and
the same half-shape are shared across the six regions.  At the transition
from level~$3$ to level~$2$, children are indexed separately by parent,
region, and half-shape.  These choices give $153$, $1\,785$, and
$41\,310$ stored occurrences at levels $4$, $3$, and $2$, respectively.
Unfolding the shared references produces an admissible recursive
construction; reuse of a parameter does not identify tensor variables
between distinct occurrences.

Including distribution logits, leaf parameters, region logits, and
entropy-dual parameters, the search has $212\,820$ real coordinates.
This is a coordinate count, not the dimension of the feasible set:
softmax gauge freedoms and the paired representation in
\eqref{eq:complement-symmetry} introduce redundancies.  The final witness
is expanded into ordinary ordered ternary strings, so its verification
does not rely on the orbit representation.

\subsection{Certified entropy penalties}
\label{sec:entropy-duals}

The maximum-entropy penalties can be bounded without certifying the
solution of an auxiliary optimization problem.  Let $\alpha$ be a
distribution on an allowed set of shapes $D$, and let
$H_D^{\max}(\alpha)$ be the maximum entropy of a distribution on $D$
having the same three marginals as $\alpha$.  Choose arbitrary strictly
positive factors $f_X,f_Y,f_Z$, and put
\begin{equation}
 t_u=\prod_{w\in\{X,Y,Z\}}f_w(u_w),\qquad
 U_D(\alpha;f)
   =\log_2\left(\sum_{u\in D}t_u\right)
                       -\sum_{u\in D}\alpha(u)\log_2 t_u.
 \label{eq:dual-upper-bound}
\end{equation}
Then $H_D^{\max}(\alpha)\leq U_D(\alpha;f)$.  Indeed, if $\rho$ has
the prescribed marginals, then
\[
 \sum_u\rho(u)\log_2t_u=\sum_u\alpha(u)\log_2t_u.
\]
Nonnegativity of the relative entropy of $\rho$ with respect to the
distribution $t/\sum_u t_u$ proves the inequality.  Therefore
\begin{equation}
 \widehat P_D(\alpha;f)=U_D(\alpha;f)-H(\alpha)
 \label{eq:certified-penalty}
\end{equation}
is a valid upper bound on the combination-loss penalty.  Its use can
only reduce the retained exponent supplied by the construction.  The
factors are optimized along with the other parameters, but their
optimality is not required for the certificate.

For completeness, the local quantities evaluated by the implementation
can be written as follows.  Extend the notation $H$ to a nonnegative
vector by $H(v)=-\sum_t v_t\log_2v_t$, with $0\log_2 0=0$.
For child weights $a_u$, distributions $\beta_u^w$, and a set $S$ of
exceptional child shapes, define
\begin{align}
 \eta_w(a,\beta;S)
 &=\sum_{u\in S}a_uH(\beta_u^w)
   +\sum_t\left[H(v_t)-H\left(\left[\sum_Lv_t(L)\right]\right)\right],
 \label{eq:eta-definition}\\
 v_t&=\sum_{\substack{u\notin S\\u_w=t}}a_u\beta_u^w.\nonumber
\end{align}
Here $H([b])=-b\log_2b$; hence each bracket in
\eqref{eq:eta-definition} is the mass of $v_t$ times the entropy of its
normalization.  At a positive occurrence $v$ in region $r$, write
$(x,y,z)=(p_r(X),p_r(Y),p_r(Z))$ and
$a_u=\alpha_{v,r}(u)+\alpha_{v,r}(s-u)$.  The three local retained
quantities are
\begin{align}
 c_{v,r,X}&=H((\alpha_{v,r})_x)
                       -\widehat P_{D_s}(\alpha_{v,r};f_{v,r}),\nonumber\\
 c_{v,r,Y}&=H(B_{v,r}^y)
                  -\eta_y(a,\beta;\{u:u_z=0\}),\label{eq:local-counts}\\
 c_{v,r,Z}&=H(B_{v,r}^z)
                  -\eta_z(a,\beta;\{u:u_x=0\text{ or }u_y=0\}).\nonumber
\end{align}
The subscripts $X,Y,Z$ on $c$ denote extraction roles, rather than
physical coordinates.

Let $m_v$ be the mass of an occurrence.  Root occurrences have mass
$m_v=\alpha^*(s_v)$, and masses are propagated by
\begin{equation}
 m_{c(v,r,u)}\mathrel{+}=
 m_v A_{v,r}\bigl(\alpha_{v,r}(u)+\alpha_{v,r}(s_v-u)\bigr).
 \label{eq:mass-recursion}
\end{equation}
Contributions from all references to a shared child are added.
The stage role totals are
\begin{equation}
 R_{\ell,k}=\sum_{v\text{ positive at level }\ell}
                    m_v\sum_{r=1}^6 A_{v,r}c_{v,r,k},
 \qquad \ell\in\{4,3\},\quad k\in\{X,Y,Z\}.
 \label{eq:stage-role-totals}
\end{equation}
At the canonical root, use $D=\{s:s_X+s_Y+s_Z=16\}$ in
\eqref{eq:certified-penalty}.  If
$\overline\beta^w=\sum_s\alpha^*(s)\beta_s^w$, the analogous root
totals are
\begin{align}
 R_{0,X}&=H(\alpha_X^*)-\widehat P_D(\alpha^*;f^*),\nonumber\\
 R_{0,Y}&=H(\overline\beta^Y)
              -\eta_Y(\alpha^*,\beta;\{s:s_Z=0\}),\label{eq:root-counts}\\
 R_{0,Z}&=H(\overline\beta^Z)
              -\eta_Z(\alpha^*,\beta;\{s:s_X=0\text{ or }s_Y=0\}).\nonumber
\end{align}

\subsection{Six-copy symmetrization and the exponent objective}

The complete construction consists of six globally coordinate-permuted
copies of the canonical parameter tree, each assigned weight $1/6$.
This symmetrization is performed in addition to the six local extraction
regions.  We write $\widehat E$ for the conservative retained exponent obtained
after inserting the entropy upper bounds
\eqref{eq:dual-upper-bound}; it need not equal the exponent obtained
with exact maximum-entropy penalties.  The six-copy symmetry yields
the following exact reduction of the stage formula for this $\widehat E$:
\begin{equation}
 \widehat E=\min_k R_{0,k}+\min_k R_{4,k}+\min_k R_{3,k}+E_2.
 \label{eq:retained-exponent}
\end{equation}
To see why the minima have this form, let $K_{r,w}$ be the canonical
contribution in physical coordinate $w$ from local region $r$, including
all occurrence masses and region weights.  For a fixed physical
orientation $h$ and a local permutation $p_r$, precisely one global
copy contributes, namely $g=h\circ p_r^{-1}$.  Its contribution in
physical coordinate $h(k)$ is $K_{r,p_r(k)}/6$.  Thus the three totals
at orientation $h$ are $R_{\ell,k}/6$, and summing the minimum over the
six physical orientations gives $\min_kR_{\ell,k}$.  This argument
uses the particular six-copy symmetry; it does not interchange a
minimum and a sum for arbitrary inputs.

The same symmetry makes the three terminal dimension exponents equal.
The contribution of a positive level-$2$ occurrence with parameter
$\mu_v$ is
\begin{equation}
 E_{2,v}=\frac{m_v}{3}
           \bigl(2+H(\mu_v,1-2\mu_v,\mu_v)\bigr),\qquad
 M_v=\frac{m_v}{3}(2-2\mu_v)\log_2q.
 \label{eq:positive-leaf}
\end{equation}
A zero-coordinate occurrence contributes no term to $E_2$ and contributes
\begin{equation}
 M_v=\frac{m_v}{3}
       \left(H(\beta_v^w)
                +\sum_L\beta_v^w(L)\,\#_1(L)\log_2q\right)
 \label{eq:zero-leaf}
\end{equation}
to $M$, where $w$ is any nonzero coordinate and $\#_1(L)$ counts the
entries equal to $1$.  Complementation shows that either nonzero
coordinate gives the same expression. If two coordinates vanish,
all distributions are point masses and the matrix-size contribution
is zero. Summing these terms gives
$E_2$ and the common matrix-size exponent $M>0$.

The degeneration theorem and asymptotic sum inequality imply
\begin{equation}
 \omega\leq\widehat\Omega(x):=\frac{C-\widehat E(x)}{M(x)}.
 \label{eq:exponent-objective}
\end{equation}
Equivalently, to prove $\omega<\tau$, it suffices to exhibit admissible
parameters for which
\begin{equation}
 \widehat E(x)+\tau M(x)-C>0.
 \label{eq:strict-margin}
\end{equation}
All three minima in \eqref{eq:retained-exponent} are retained in the
certified evaluation.

\subsection{From the construction to a multiplication algorithm}
\label{sec:algorithm}

The rational parameter assignment specifies the retained constituents,
their multiplicities, and their matrix dimensions in the published
tensor construction. For increasing tensor powers, the extraction
theorem yields independent matrix multiplication tensors with normalized
logarithmic multiplicity at least $\widehat E$, common normalized
logarithmic dimension $M$, and input rank cost at most $C$, up to
vanishing losses in the normalized exponents.
The six global coordinate permutations balance the final dimensions.
The asymptotic sum inequality then gives
$\omega\leq\widehat\Omega=(C-\widehat E)/M$.

For every $\varepsilon>0$, the standard conversion from these tensor
bounds to algorithms yields multiplication of two $n\times n$ matrices
in $O(n^{\widehat\Omega+\varepsilon})$ arithmetic operations.
The rational parameters are a concrete admissible input to the
degeneration and recursive blocking machinery of
\cite{alman2025,schonhage1981}; the entropy-dual factors certify the
analysis of the retained constituents.

Since the exact certificate proves $\widehat\Omega<2.37115924$, choose
$0<\varepsilon<2.37115924-\widehat\Omega$. This yields the
$O(n^{2.37115924})$ algorithm stated in Theorem~\ref{thm:main}.
The improvement is carried by the new recursive extraction parameters
and their certified analysis. The conversion from that construction
to a matrix multiplication algorithm uses the established tensor
framework.

\section{Structured optimization of the construction}
\label{sec:optimization}

The optimization procedure is organized around three structural features
of the construction: occurrence masses span different scales, the
retained exponent is governed by competing role minima, and the final
parameters must admit rigorous certification. We address these with
coordinate conditioning, a separate bottleneck correction, and positive
rational export. The orbit representation and equality-constrained
starting approach were informed by computational notes shared by Eric
Naslund; the following sections specify the implementation used here.
The documented procedure differs from the Adam and differentiated
Sinkhorn approach described in \cite[Section~3]{dupont2026} by combining
conditioned quasi-Newton solves with a separate six-control correction.

\subsection{Smooth parameterization and balanced search objective}

Probability vectors are represented by softmax logits.  For a local
split distribution, a raw logit $z_u$ is replaced by
$(z_u+z_{s-u})/2$ before applying the softmax, imposing
\eqref{eq:complement-symmetry}.  For an orbit of size $b$, adding
$\ln b$ before normalization converts a per-string logit into an
orbit mass.  Leaf parameters use $\mu=\tfrac12\operatorname{sigmoid}(z)$.
Entropy-dual factors are represented by unconstrained logarithms.
These transformations enforce the local probability constraints during
optimization, while the exported witness enforces them exactly.

Near the selected point, the three roles are nearly tied at the root
and at each constituent stage.  We exploit this structure in the
numerical search by defining, for $t\in\{0,4,3\}$,
\[
 \overline R_t=\frac13\sum_kR_{t,k},\qquad
 h_{t,Y}=R_{t,Y}-R_{t,X},\qquad
 h_{t,Z}=R_{t,Z}-R_{t,X}.
\]
The six quantities $h$ are equality residuals. The difference between
the average-based retained exponent and the actual one is exactly
\begin{equation}
\mathcal B(x)=\sum_{t\in\{0,4,3\}}
  \left(\overline R_t(x)-\min_k R_{t,k}(x)\right)\geq0.
\label{eq:bottleneck-gap}
\end{equation}
Thus an increase in the average can be lost to greater role imbalance.
The condition $h=0$ makes $\mathcal B=0$; it identifies the six
differences targeted by the correction below.  With target parameter
$\tau_0=2.37116$, the augmented-Lagrangian objective minimized during
the final continuation is
\begin{equation}
 \mathcal L(x;\lambda,\rho)
 =C-E_2(x)-\sum_{t\in\{0,4,3\}}\overline R_t(x)-\tau_0M(x)
       +\langle\lambda,h(x)\rangle
       +\frac{\rho}{2}\|h(x)\|_2^2.
 \label{eq:augmented-objective}
\end{equation}
The average-role expression is a search surrogate.  It agrees with
the retained exponent only when $h=0$; otherwise it can overestimate
the retained exponent.  No conclusion about $\omega$ is drawn from
this surrogate alone, and equality of the roles is not a hypothesis
of the certified construction.

After each outer pass, the multipliers are updated by
$\lambda\leftarrow\lambda+\rho h$.  The penalty starts at $\rho=20$
and is doubled after every fourth pass if
$\|h\|_\infty>2\cdot10^{-7}$.  The final continuation starts from a
previously obtained point with minimum-based objective approximately
$2.371161129895562$, taking multipliers from the preceding continuation
and resetting the penalty to $20$.  It consists of $20$ outer passes, each with an inner
iteration limit of $500$, and records $10\,087$ objective-and-gradient
evaluations.  The numerical computations use double precision,
automatic differentiation in JAX, and SciPy's L-BFGS-B implementation~\cite{byrd1995}, without explicit
bounds on the transformed coordinates.

\subsection{Mass-dependent conditioning of quasi-Newton steps}

The occurrence masses and local probabilities vary substantially in
magnitude. The source of the scale mismatch can already be seen in a
local term $m\,g(p)$ with softmax probabilities $p$, holding its incoming
mass $m$ fixed: its logit gradient is
\[
\frac{\partial (m\,g(p))}{\partial z_i}
=m p_i\left(\frac{\partial g}{\partial p_i}
       -\sum_jp_j\frac{\partial g}{\partial p_j}\right).
\]
Small masses, and often small probabilities, suppress these coordinate gradients.
This motivates a diagonal rescaling using mass and probability weights,
also suggested by the corresponding Fisher-information scales.  At the beginning of an outer pass,
let $x^{(0)}$ be the current parameter vector.  Optimize in coordinates
$y$ defined by
\begin{equation}
 x=x^{(0)}+S y,\qquad
 S_{ii}=\operatorname{clip}_{[0.03,10^6]}
       \left(\max\{w_i,10^{-14}\}^{-1/2}\right).
 \label{eq:coordinate-scaling}
\end{equation}
The weight $w_i$ is the current occurrence mass times the appropriate
local probability: orbit mass for a free complete split distribution,
region weight for region logits, and region weight times split
probability for split logits.  Root logits use their root probabilities.
For a dual marginal coordinate of marginal probability $p$, the weight
is the occurrence and region mass times $p(1-p)$.  A leaf-$\mu$ logit
uses its occurrence mass.  This is a diagonal conditioning rule, not
an exact natural-gradient method.  The scale is held fixed throughout
each inner solve and recomputed between outer passes.

Each inner solve uses L-BFGS history length $40$, function tolerance
$10^{-15}$, gradient tolerance $10^{-9}$, and a line-search limit of
$35$.  Relative to the preceding continuation, the scale cap is
increased from $10^5$ to $10^6$, the mass floor is decreased from
$10^{-12}$ to $10^{-14}$, and the history length is increased from
$30$ to $40$.  These changes concern numerical conditioning only;
they do not alter the admissible tensor family.

\subsection{Balancing extraction bottlenecks with six controls}

The bottleneck gap in \eqref{eq:bottleneck-gap} motivates a correction
that acts directly on the six role differences, even though the full
parameter vector has $212\,820$ coordinates. After every second outer
pass, this separate correction attempts to reduce
the six role-equality residuals using only six scalar controls.  Two
controls add to the root logits the features
\[
 \left(\frac{s_X-s_Y}{16},\frac{s_Y-s_Z}{16}\right).
\]
At each level $\ell\in\{4,3\}$, two controls add to the region logits
at shape $s$ and permutation $p_r$ the features
\[
 \left(\frac{s_{p_r(Y)}-s_{p_r(X)}}{2^\ell},
       \frac{s_{p_r(Z)}-s_{p_r(X)}}{2^\ell}\right).
\]
A $6\times6$ Jacobian is estimated by central differences with step
$10^{-4}$.  Up to four Newton corrections are computed by least
squares with relative singular-value cutoff $10^{-12}$, scaling a
correction down if its infinity norm exceeds $1$.  The balancing
attempt terminates early when $\|h\|_\infty<10^{-11}$.

Both unbalanced iterates and balanced candidates are evaluated using
the minimum-based objective \eqref{eq:exponent-objective}, and the best
candidate is retained.  Balancing is a candidate-generation step:
the augmented-Lagrangian continuation proceeds from its own iterate.
In floating-point monitoring, a minimum of three entries is evaluated
with the conservative approximation
\[
 \operatorname{smin}_\varepsilon(a)
       =-\varepsilon\ln\sum_{k=1}^3 e^{-a_k/\varepsilon},
 \qquad \varepsilon=10^{-14},
\]
which lies between $\min_k a_k-\varepsilon\ln3$ and $\min_k a_k$.
The final rational evaluation uses the actual minima.

The selected point comes from the balancing step following the last
outer pass.  Its floating-point minimum-based objective is approximately
$2.3711592385931133$.  The search is local and supplies no certificate
of global optimality.  The reported settings describe the continuation
that produced the witness; reproduction of the mathematical bound
requires the final witness and its checker, independently of repeating
the optimization trajectory.

\subsection{Positive rational export and independent reconstruction}

The selected probabilities are converted to integer numerators with
denominator $2^{80}$, preserving normalization and the paired symmetry
of local split probabilities.  Orbit masses are then expanded into
equal probabilities for their ordinary strings.  Increasing the common
denominator to $2^{96}$ makes this expansion exact, since the relevant
orbit sizes are powers of two.  Consequently the delivered witness
contains rational distributions on explicitly ordered strings rather
than an orbit-compressed numerical model.

The logarithmic dual parameters are converted separately.  In each
factor row, subtract the smallest active logarithmic value before
exponentiation and rational approximation.  The resulting active
factors are at least $1$, and their integer numerators are not restricted
to machine-word size.  This choice preserves large ratios without forcing small factors onto
a fixed positive floor. In particular, normalizing by the largest
factor and clipping small rounded values could substantially weaken
the entropy bound despite preserving positivity. The chosen row
normalization avoids that loss. Such row rescalings leave
\eqref{eq:dual-upper-bound} unchanged.  More generally, the validity
of the bound requires only strict positivity of the exported factors;
neither a small approximation error nor a stationarity check is an
assumption of the verification.

The checker validates support, nonnegativity, exact normalization,
child shapes, complementary references, and the recursively reconstructed
complete split distributions.  It evaluates the six global copies in
physical coordinates before taking the published minima.  Binary
logarithms and subsequent expressions are enclosed by rational outward
intervals.  The numerical result is reported in Section~\ref{sec:verification}.

\section{Rigorous verification}
\label{sec:verification}

Let $E_{\mathrm{pub}}$ denote the retained exponent computed with exact
maximum-entropy penalties, and let $\widehat E$ denote the value in
\eqref{eq:retained-exponent}, using the positive-factor bounds
\eqref{eq:dual-upper-bound}. Since these bounds only increase the
penalties, and all mass weights are nonnegative,
\begin{equation}
\widehat E\leq E_{\mathrm{pub}}.
\label{eq:monotonicity}
\end{equation}
Every sum and minimum occurring after the replacement is
coordinatewise nondecreasing, which preserves this inequality.

\begin{proposition}[Certificate implication]
\label{prop:certificate}
For admissible parameters, suppose the six-copy construction gives
the conservative retained exponent $\widehat E$, common matrix-size
exponent $M>0$, and rank exponent $C$. Then
\[
\omega\leq\frac{C-\widehat E}{M}.
\]
Consequently, a rational target $\tau$ satisfying
$\widehat E+\tau M-C>0$ proves $\omega<\tau$.
\end{proposition}

\begin{proof}
The combination-loss theorem~\cite{alman2025}, in the formulation
of~\cite[Theorem~1]{dupont2026}, together with the asymptotic sum
inequality~\cite{schonhage1981}, gives
$E_{\mathrm{pub}}+M\omega\leq C$.
Apply \eqref{eq:monotonicity}, divide by $M>0$, and rearrange the
strict feasibility inequality.
\end{proof}

\subsection{Exact checks and interval evaluation}

The witness is a JSON file containing integer numerators for the
probabilities and entropy-dual factors, together with node shapes,
levels, regional data, and child references. Both declared denominators
are $2^{96}$. The checker reconstructs the distributions on ordinary
ordered ternary strings, without importing the optimizer or its orbit
representation. It checks the integer types, support, nonnegativity,
normalization, strictly positive dual factors, child levels and shapes,
complementary split indices, and leaf-parameter ranges. The directed
graph represents the recursively unfolded parameter tree; the validity
of mass accumulation under this sharing is proved in
Lemma~\ref{lem:audit-unfolding}.

All probability and mass arithmetic is exact rational arithmetic.
Every logarithm is enclosed using a positive series with an explicit
remainder bound and outward rounding on a $2^{-192}$ grid, as detailed
in Lemma~\ref{lem:audit-logarithms}. Subsequent expressions are
evaluated by rational interval operations. In particular, for
$x_i\in[a_i,b_i]$, the checker uses
\[
\min_i x_i\in[\min_i a_i,\min_i b_i].
\]
It explicitly sums the physical-coordinate contributions of all
six global copies and six local regions, and then applies the original
regional minima. This provides a second implementation of the
symmetry reduction used during optimization.

Table~\ref{tab:certificate} gives outward decimal enclosures for the
final calculation. Each printed endpoint denotes an exact terminating
rational number; the accompanying result file contains the full
rational endpoints. The interval for $\widehat\Omega$ encloses the
value of $(C-\widehat E)/M$ for this particular witness. It is not a
two-sided estimate of $\omega$ and does not establish global optimality.

\begin{table}[ht]
\centering
\small
\begin{tabular}{@{}lll@{}}
\toprule
Quantity & Lower endpoint & Upper endpoint\\
\midrule
$\widehat E$ & $8.1061402482833692887887$ & $8.1061402482833692887888$\\
$M$ & $6.0530304732690271387406$ & $6.0530304732690271387407$\\
$C$ & $22.4588393764608328595357$ & $22.4588393764608328595358$\\
$\widehat\Omega$ & $2.3711592385931075066123$ & $2.3711592385931075066124$\\
\bottomrule
\end{tabular}
\caption{Outward enclosures obtained from the rational witness.}
\label{tab:certificate}
\end{table}

At the exact rational target $\tau=2.37115924$, the checker obtains
\begin{equation}
\widehat E+\tau M-C\in
[8.5159631350886\cdot10^{-9},\;8.5159631350887\cdot10^{-9}].
\label{eq:checked-margin}
\end{equation}
The positive lower endpoint supplies the strict inequality, with
$M>0$ certified separately.

\begin{proof}[Proof of Theorem~\ref{thm:main}]
The exact checks establish admissibility of the supplied rational
witness. Apply Proposition~\ref{prop:certificate} using the positive
lower endpoint in \eqref{eq:checked-margin}. The strict gap permits
the choice of $\varepsilon$ in Section~\ref{sec:algorithm}, proving
the $O(n^{2.37115924})$ arithmetic bound. A matrix multiplication
algorithm over $\mathbb Q$ remains valid after extending scalars, giving
the asserted consequence for characteristic-zero extensions.
\end{proof}

\subsection{Reproduction and scope}

The ancillary directory contains the witness in \texttt{witness.json}
and the checker in \texttt{verify.py} and \texttt{rigorous\_interval.py}.
The recorded result is \texttt{expected\_result.json}; documentation is
included alongside these files. With Python~3.10
or later, the verification command is
\begin{verbatim}
python3 verify.py witness.json --target 2.37115924 \
    --out result.json
\end{verbatim}
The checker uses only the standard library. It requires no package
installation, model, remote service, credentials, or network access.
The exact witness bytes have SHA-256 digest
\begin{center}
\small\ttfamily
87f5e7a8c722ca05464920c99c64083e5fa\\
8235d1822720278e5a058485532ff
\end{center}
where the two lines are concatenated. The package also includes
file-level checksums. A regenerated result can be compared with
the recorded rational endpoints; elapsed time and diagnostic output
are not part of the mathematical certificate.

This is a computer-assisted proof relying on the cited tensor theorem
and the supplied finite arithmetic calculation. The verifier is
auditable source code, rather than a proof-assistant formalization of
the full tensor theorem. Verification reproduces the bound from the
final witness. It does not replay the numerical optimization trajectory,
whose last continuation inherits an earlier point and multipliers.
The optimization libraries used to discover the point are therefore
unnecessary for verifying it.

\section{Discussion}

The certified construction improves the asymptotic arithmetic
complexity of matrix multiplication within the established
combination-loss framework. Its algorithmic content is the rational
recursive extraction plan, which fixes the allocation of tensor
constituents and the tradeoff between retained multiplicity and matrix
size. The exact certificate turns that plan into a theorem about
matrix multiplication.

The computational methodology addresses three coupled difficulties:
weak gradients from rare recursive occurrences, a retained exponent
controlled by competing minima, and entropy penalties that must survive
rational export. Mass-dependent conditioning, a low-dimensional
bottleneck correction, and an independent exact reconstruction address
these difficulties at different stages. The experiment establishes the
success of the combined procedure; it does not isolate the contribution
of each component through controlled ablations.

The resulting $212\,820$-coordinate construction also shows that useful
structural restrictions can coexist with an improved exponent.
Selective removal of parameter identifications and alternative
bottleneck corrections are natural directions for further algorithmic
constructions. The certificate provides a fixed, independently
checkable reference point for evaluating such extensions.

\paragraph{Acknowledgments and computational assistance.}
The initial exploratory implementation was informed by publicly
circulated computational notes shared by Eric Naslund in 2026.
The numerical implementation was reconstructed and extended, and the
resulting witness was evaluated with the independently implemented
ordinary-string verifier accompanying this paper.
Generative AI tools assisted the numerical exploration, software
development, and preparation of the manuscript. The numerical certificate
and verification software are supplied for independent reproduction.

\bibliographystyle{plainnat}
\bibliography{references}

\appendix
\section{Arithmetic and representation lemmas}

\subsection{Shared recursive parameters}

\begin{lemma}[Unfolding shared parameters]
\label{lem:audit-unfolding}
Suppose a finite layered directed acyclic graph specifies the parameters
of the combination-loss construction. Assume that every child edge has
the prescribed lower level and shape, and that all distributions and
leaf parameters satisfy their exact support, normalization, and range
conditions. The graph defines an admissible parameter tree by giving
each occurrence of a referenced node a separate copy of its parameters
and descendants. Its derived distributions and exponent contributions
can be evaluated on the graph, accumulating the masses of repeated
occurrences before computing contributions linear in those masses.
\end{lemma}

\begin{proof}
Unfold the graph recursively. Its levels strictly decrease, so the
resulting tree is finite. An induction from the leaves proves that all
copies of a graph node have identical complete split distributions:
the recursion uses only its parameters and its children's distributions.
Its local entropies and matrix-size coefficients are therefore also
identical. Mass propagation, regional sums before minimization, and leaf
contributions are linear in the incoming mass. Summing the masses of
these copies consequently gives the same result as their separate
evaluation. No tensor variables are identified by this parameter sharing.
\end{proof}

\subsection{Rational logarithm enclosures}
\label{sec:audit-logarithms}

\begin{lemma}[Outward dyadic logarithms]
\label{lem:audit-logarithms}
The logarithm routine in the accompanying verifier returns rational
numbers $L_-(x)\leq\log_2x\leq L_+(x)$ for every positive rational
input $x$. The enclosure uses only integer arithmetic.
\end{lemma}

\begin{proof}
Write $x=2^kz$ with integer $k$ and $1\leq z<2$, using exact comparisons
of integers. Let $S=2^{192}$, $m=74$, and
$t=(z-1)/(z+1)\in[0,1/3)$. The identity
\[
 \ln z=2\sum_{j=0}^{m-1}\frac{t^{2j+1}}{2j+1}+\mathcal R_m,
 \qquad
 0\leq\mathcal R_m\leq
 \frac{2t^{2m+1}}{(2m+1)(1-t^2)}
\]
follows by integrating the geometric series, and bounding every
denominator in its remaining positive terms by $2m+1$.

The implementation brackets $t$ and $t^2$ on the grid $S^{-1}\mathbb Z$.
If $P_j^-/S$ and $P_j^+/S$ bracket $t^{2j+1}$, and $V^-/S,V^+/S$
bracket $t^2$, the integer updates
\[
 P_{j+1}^-=\left\lfloor\frac{P_j^-V^-}{S}\right\rfloor,
 \qquad
 P_{j+1}^+=\left\lceil\frac{P_j^+V^+}{S}\right\rceil
\]
preserve this invariant. Each summand is divided by $2j+1$ with
outward rounding. After $m$ summands, the upper tail in units of
$S^{-1}$ is bounded by
\[
 \left\lceil
 \frac{2P_m^+S}{(2m+1)(S-V^+)}
 \right\rceil.
\]
The denominator is positive since $t\leq1/3$ and the grid bounds give
$V^+<S$. This produces integers $A_-,A_+$ enclosing $S\ln z$.
The same calculation at $z=2$ produces integers $B_-,B_+$ enclosing
$S\ln2$, with $B_->0$. Therefore
\[
 k+\frac{1}{S}\left\lfloor\frac{SA_-}{B_+}\right\rfloor
 \;\leq\;\log_2x\;\leq\;
 k+\frac{1}{S}\left\lceil\frac{SA_+}{B_-}\right\rceil.
\]
If $z=1$, the routine returns the exact integer $k$.
\end{proof}

All subsequent interval operations are exact rational operations with
outward endpoints. In particular, if $x_i\in[a_i,b_i]$, then
\[
 \min_i x_i\in[\min_i a_i,\min_i b_i].
\]
For a nonnegative rational vector $v$ of mass $s$, define
$\mathcal H(v)=-\sum_i v_i\log_2v_i$, with $0\log_20=0$.
The compatibility calculation uses
$\mathcal H(v)-\mathcal H((s))=sH(v/s)$ when $s>0$, and zero when
$s=0$. This identity also covers mixtures whose mass exceeds one.

\end{document}